\documentclass{article}
\usepackage{fullpage}
\usepackage{amsmath}
\usepackage{amsthm}
\usepackage{amssymb}
\usepackage{url}
\usepackage{graphicx}
\usepackage{verbatim}
\usepackage{url}
\usepackage{graphicx}
\usepackage{tikz}
\usepackage[noline,boxed,figure]{algorithm2e}
\usepackage{array}
\usepackage{epsfig}
\usepackage{subcaption}
\usetikzlibrary{calc}
\usetikzlibrary{patterns}
\tikzstyle{mybox} = [draw=black, fill=white,  thick,
    rectangle, inner sep=10pt, inner ysep=20pt]
\tikzstyle{mybox} = [draw=black, fill=white,  thick,
    rectangle, inner sep=2pt, inner ysep=2pt]

\newtheorem{thm}{Theorem}
\newtheorem{lemma}{Lemma}

\newtheorem{prop}{Proposition}
\theoremstyle{definition}
\newtheorem{definition}{Definition}
\newtheorem{remark}{Remark}

\newtheorem{hypothesis}{Hypothesis}

\begin{document}

%\title{A Geometric Algorithms for Solving a Linear System}
%\title{New Iterative Methods for Solving Linear Systems Via The Triangle Algorithm}
%\title{Iterative Methods for Solving Linear Systems Via The Convex Hull Triangle Algorithm}
\title{Randomized Triangle Algorithms for Convex Hull Membership}
%\title{Working Paper}
\author{ Bahman Kalantari\\
Department of Computer Science, Rutgers University, NJ\\
kalantari@cs.rutgers.edu
}
\date{}
\maketitle

\begin{abstract}
We present randomized versions of the {\it triangle algorithm} introduced in \cite{kal14}. The triangle algorithm tests membership of a distinguished point $p \in \mathbb{R} ^m$  in the convex hull of a given set $S$ of $n$ points in $\mathbb{R}^m$.  Given any {\it iterate} $p' \in conv(S)$, it searches for a {\it pivot}, a point $v \in S$ so that $d(p',v) \geq d(p,v)$.  It replaces $p'$ with the point on the line segment $p'v$ closest to $p$ and repeats this process. If a pivot does not exist, $p'$ certifies that
$p \not \in conv(S)$.  Here we propose two  random variations of the triangle algorithm that allow relaxed steps so as to take more effective steps possible in subsequent iterations. One is inspired by the {\it chaos game} known to result in the Sierpinski triangle.  The incentive is that randomized iterates together with a property of Sierpinski triangle would result in effective pivots. Bounds on their expected complexity coincides with those of the deterministic version derived in \cite{kal14}.
\end{abstract}

%\newpage

\section{Introduction}

Given a finite set $S= \{v_1, \dots, v_n\} \subset \mathbb{R} ^m$, and a distinguished point $p \in \mathbb{R} ^m$, the {\it convex hull membership problem} (or {\it convex hull decision problem}) is to test if $p \in conv(S)$, the convex hull of $S$. Given a desired tolerance $\varepsilon \in (0,1)$, we call a point $p_\varepsilon \in conv(S)$ an $\varepsilon$-approximate solution if $d(p_\varepsilon,p) \leq \varepsilon R$, where $R= \max \{d(p,v_i): \quad i=1, \dots, n\}$. The convex hull membership problem is the most basic of the convex hull problems, see \cite{Goodman} for general convex hull problems. Nevertheless, it is a fundamental problem in computational geometry and linear programming and finds applications in statistics, approximation theory, and machine learning. Problems related to the convex hull membership include, computing the distance from a point to the convex hull of a finite point set,
support vector machines (SVM), approximating functions as convex combinations of other functions, see e.g. Clarkson \cite{clark2008} and Zhang \cite{zhang}, and \cite{kal14}. From the theoretical point of view the problem is solvable in polynomial time via the pioneering algorithm of Khachiyan  \cite{kha79}, or  Karmarkar \cite{kar84}.  For  large-scale problems greedy algorithms are preferable to polynomial-time algorithms. The best known such algorithms are, Frank-Wolfe algorithm \cite{Frank}, Gilbert's algorithm \cite{Gilbert}, and {\it sparse greedy approximation}. For connections between these and analysis  see  Clarkson \cite{clark2008},  G{\"a}rtner and Jaggi \cite{Gartner}.

A recent algorithm for the convex hull membership  problem is the {\it triangle algorithm} \cite{kal14}. It can either compute an $\varepsilon$-approximate solution, or when $p \not \in conv(S)$ a separating hyperplane and a point that approximates the distance from $p$ to $conv(S)$ to within a factor of $2$.
Based on preliminary experiments, the triangle algorithm performs quite well on reasonably large size problem, see \cite{Meng}. It can also be applied to solving linear systems, see \cite{kal12a} and \cite{Gibson} (for experimental results). Additionally, it can be applied to linear programming, see \cite{kal14}. Some variations of the triangle algorithm are given in \cite{kal12a} and \cite{kalSaks}.  The performance of the triangle algorithm is quite fast in detecting the cases when $p$ is not near a boundary point of $conv(S)$.  When $p$ is a near-boundary point of $conv(S)$ the triangle algorithm  may experience zig-zagging in achieving high accuracy approximations.  In \cite{kal14} we have described several strategies to remedy this, such as adding new auxiliary points to $S$.  In this article  we propose two randomized versions of the triangle algorithm. The  randomized algorithms are also applicable to solving linear systems and linear programming.

The article is organized as follows. In Section \ref{sec2}, we review the triangle algorithm, its relevant properties as well as bounds on its worst-case time complexities. In Section \ref{sec3}, we describe a randomized version, called {\it Greedy-Randomized Triangle Algorithm}. In Section \ref{sec4}, we describe a second randomized triangle algorithm inspired by the chaos game, see  Barnsley \cite{Barn93} and Devaney \cite{Devaney2004}, known to give rise to the well-known Sierpinski triangle.  We call this algorithm {\it Sierpinski-Randomized Triangle Algorithm}. We conclude with some remarks.

\section{Review of The Triangle Algorithm} \label{sec2}

Here we review the terminology and some results from \cite{kal14}. The Euclidean distance is denoted by $d(\cdot, \cdot)$.

\begin{definition} Given $p' \in conv(S)$, we say $v \in S$ is a {\it pivot} relative to $p$  at $p'$ (or $p$-pivot, or simply pivot) if $d(p',v) \geq d(p,v)$ (see Figure \ref{Fig1}).
\end{definition}

\begin{definition} Given $p' \in conv(S)$, we say $v \in S$ is a {\it strict pivot} relative to $p$ at $p'$ (or {\it strict} $p$-pivot, or simply {\it strict} pivot) if $\theta=\angle p'pv \geq \pi/2$ (see Figure \ref{Fig1}).
\end{definition}

\begin{definition} \label{defn4} We call a point $p' \in conv(S)$ a $p$-{\it witness} (or simply a {\it witness}) if $d(p',v_i) < d(p,v_i)$, for all $i=1, \dots, n$.
\end{definition}

A witness has the property that the orthogonal bisecting hyperplane to the line $pp'$ separates $p$ from $conv(S)$. Furthermore,
$$\frac{1}{2}\leq d(p,p')  \leq  d(p, conv(S)) \leq d(p,p').$$

\begin{thm} \label{thm1} {\bf (Distance Duality \cite{kal14})} $p \in conv(S)$ if and only if given any  $p' \in conv(S)$, there exists a pivot.
\end{thm}

\begin{thm}  \label{thm2} {\bf (Strict Distance Duality \cite{kal14})} Assume $p \not \in S$.  Then $p \in conv(S)$ if and only if given any  $p' \in conv(S)$, there exists a strict pivot.
\end{thm}

\begin{definition} Given three points $p,p',v \in \mathbb{R}^m$ such that $d(p',v) \geq d(p,v)$. Let $nearest(p; p'v)$ be the nearest point to $p$ on the line segment joining  $p'$ to $v$. Specifically, let
\begin{equation}
\alpha = \frac{(p-p')^T(v-p')}{d^2(v, p')}.
\end{equation}
Then
\begin{equation} \label{pdp}
nearest(p; p'v)=
\begin{cases}
(1-\alpha)p' + \alpha v, &\text{if $\alpha \in [0,1]$;}\\
v, &\text{otherwise.}
\end{cases}
\end{equation}
\end{definition}

\begin{remark}
By squaring the distances we have
\begin{equation}
d(p',v) \geq d(p,v) \iff  p'^Tp' - p^Tp \geq 2v^T(p'-p).
\end{equation}
Thus to search for a pivot does not require taking square-roots.
Neither does the computation of $nearest(p;p'v)$. It requires $O(mn)$ arithmetic operations.
\end{remark}
The triangle algorithm is summarized in the box.

\begin{algorithm}[htpb]
{\bf Triangle Algorithm}\

 \KwIn {$S= \{v_1, \dots, v_n\}$, $p$, $\varepsilon \in (0,1)$}
  \KwOut {$p' \in conv(S)$, either $d(p, p') \leq \varepsilon R$,
 or  $p'$ is a {\bf Witness}}
 $p'={\rm argmin}\{d(p,v): v \in S\}$\;
 \While{$(d(p,p')> \varepsilon R)$}{
  \eIf{no pivot exists}{
   Output $p'$ is a {\bf Witness} and halt\;
   }{given a pivot $v$, set
   $p' =nearest(p;p'v)$\;
  }
 }
 Output $p'$\;
%\caption{Triangle algorithms} \label{fig:Triangle}
\end{algorithm}

\begin{thm}  \label{thm3} {\rm (\cite{kal14})} Given $\varepsilon \in (0,1)$, if $p \in conv(S)$, the number of arithmetic operations of the triangle algorithm to compute $p_\varepsilon$ so that $d(p, p_\varepsilon) \leq \varepsilon R$ is
$$O\bigg (\frac{mn}{\varepsilon^2} \bigg ).$$
\end{thm}

\begin{figure}[htpb]
	\centering
	\begin{tikzpicture}[scale=0.6]
%\draw (0, 10) rectangle (10, 0);	

%\draw (0, 10) rectangle (10, 0);	

\begin{scope}[red]
        % \draw (0.0,0.0) circle (7.0);
		 %\draw (7.0,0.0) circle (7.0);
		 %\draw (0.0,0.0) circle (4.48);
\end{scope}
		
		\draw (0.0,0.0) -- (7.0,0.0) -- (-2.0,-4.0) -- cycle;
      \draw (0,0) -- (7,0) node[pos=0.5, above] {};
      \draw (-2,-4) -- (0,0) node[pos=0.5, above] {};
       \draw (0,0) -- (1.15,-2.6) node[pos=0.5, right] {};
       \draw (1.15,-2.6) node[below] {$p''$};
       \filldraw (1.15,-2.6) circle (2pt);
		\draw (0,0) node[left] {$p$};
		\draw (7,0) node[right] {$v$};
		\draw (-2,-4) node[below] {$p'$};
\draw (-1.5,-3.5) node[above] {$~~~~\theta$};
         %\draw (14,0) node[right]{$C'$};
          %\draw (-7,0) node[left]{$C$};
           %\draw (-4.48,0) node[left]{$C''$};
           \filldraw (0,0) circle (2pt);
\filldraw (7,0) circle (2pt);
\filldraw (-2,-4) circle (2pt);
		
		%\filldraw (-2,7) circle (1pt) node[left] {$p$};
	\end{tikzpicture}
\begin{center}
\caption{An example of an iterate, a strict pivot, and $p''=nearest(p;p'v)$.} \label{Fig1}
\end{center}
\end{figure}
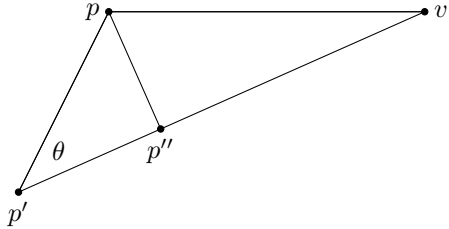

\begin{thm}  \label{thm4} {\rm (\cite{kal14})}  Assume $p$ lies in the relative interior of $conv(S)$. Let $\rho$ be the supremum of radii of the balls centered at $p$ in this relative interior. Given $\varepsilon \in (0,1)$, suppose the triangle algorithm uses a strict pivot in each iteration. The number of arithmetic operations to compute $p_\varepsilon \in conv(S)$ so that $d(p_\varepsilon,p) < \varepsilon R$ is
$$O\bigg (mn \bigg (\frac{R}{\rho} \bigg)^2 \ln \frac{1}{\varepsilon} \bigg).$$
\end{thm}

\section{Greedy-Randomized Triangle Algorithm} \label{sec3}

In this section we describe a randomized algorithm we call {\it Greedy-Randomized Triangle Algorithm}.  It is designed to avoid  possible zig-zagging in the triangle algorithm.  Given an iterate $p'$, it computes a pivot $v$, if it exists. Then it randomly selects the new iterate as the midpoint of $p'$ and $v$, or  $nearest(p; p'v)$. It records the closest known point to $p$ as $p_*$, the {\it current incumbent candidate}, and updates it whenever necessary.

\begin{algorithm}[htpb]
{\bf Greedy-Randomized Triangle Algorithm}\

 \KwIn {$S= \{v_1, \dots, v_n\}$, $p$, $\varepsilon \in (0,1)$}
 \KwOut {$p' \in conv(S)$, either $d(p, p') \leq \varepsilon R$,
 or  $p'$  is a {\bf Witness}}
 $p'={\rm argmin}\{d(p,v): v \in S\}$, $p_*=p'$\;
 \While{$(d(p,p_*)> \varepsilon R)$}
{
  \eIf{no pivot exists at $p'$}{Output $p'$ is a {\bf Witness} and halt\;}
  {given a pivot $v$,
   randomly set $p'=(p'+v)/2$, or $p'=nearest(p;p'v)$\;
   Update $p_*$;
  }
}
 Output $p'=p_*$\;
 %\caption{Greedy Randomized Triangle algorithms} \label{fig:Triangle}
\end{algorithm}

Figure \ref{Fig2} describes a case where given an iterate $p'$ and pivot $v_1$, we can get closer to $p$ by selecting  the closest point $p''$ on $p'v_1$.   However, selecting instead  $p'''$, the midpoint of $p'$ and $v_1$, we create the chance to select a better approximation using $p'''$ as iterate.

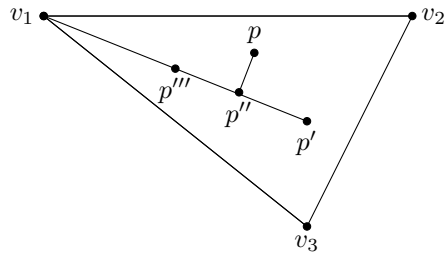
\begin{figure}[htpb]
	\centering
	\begin{tikzpicture}[scale=0.7]

		\draw (0.0,0.0) -- (7.0,0.0) -- (5.0,-4.0) -- cycle;
      \draw (0,0) -- (7,0) node[pos=0.5, above] {};
      \draw (5,-4) -- (0,0) node[pos=0.5, above] {};
       %\draw (0,0) -- (1.15,-2.6) node[pos=0.5, right] {};
       \draw (4,-.7) node[above] {$p$};
       %\filldraw (1.15,-2.6) circle (2pt);
		\draw (0,0) node[left] {$v_1$};
		\draw (7,0) node[right] {$v_2$};
		\draw (5,-4) node[below] {$v_3$};
\draw (5,-2) node[below] {$p'$};
%\draw (-1.5,-3.5) node[above] {$~~~~\theta$};
           \filldraw (0,0) circle (2pt);
\filldraw (7,0) circle (2pt);
\filldraw (5,-4) circle (2pt);
\filldraw (0,0) circle (2pt);
\filldraw (4,-.7) circle (2pt);
\filldraw (5,-2) circle (2pt);
\draw (5,-2) -- (0,0) node[pos=0.5, above] {};
\filldraw (2.5,-1) circle (2pt);
\filldraw (3.71,-1.45) circle (2pt);
\draw (3.71,-1.45) -- (4,-.7) node[pos=0.5, above] {};
\draw (3.71,-1.45) node[below] {$p''$};
\draw (2.5,-1) node[below] {$p'''$};
		
		%\filldraw (-2,7) circle (1pt) node[left] {$p$};
	\end{tikzpicture}
\begin{center}
\caption{An example where $p'''=(p'+v_1)/2$ is a better iterate than $p''=nearest(p;p'v_1)$ for the next iteration.} \label{Fig2}
\end{center}
\end{figure}

From properties of the triangle algorithm  reviewed  in the previous section we have,

\begin{thm} If $p \in conv(S)$, bound on the expected number of arithmetic operations of the Greedy-Randomized Triangle Algorithm to compute an $\varepsilon$-approximate solution is
$$O \bigg (\frac{mn}{\varepsilon^2} \bigg ).$$
Moreover, if it is known that $p$ is the center of ball of radius $\rho$ contained in the relative interior of $conv(S)$, and if each times it computes a pivot for an iterate the pivot is a strict pivot, then bound on the expected number of arithmetic operations to compute an $\varepsilon$-approximate solution is
$$O\bigg (mn \bigg (\frac{R}{\rho} \bigg)^2 \ln \frac{1}{\varepsilon} \bigg). ~\qed$$
\end{thm}

\section{A Randomized Triangle Algorithm Based on The Chaos Game} \label{sec4}

As described by Devaney \cite{Devaney2004}:

{\it The {\it chaos game} and its multitude of variations provides a wonderful opportunity to combine elementary ideas from geometry, linear algebra, probability, and topology with some quite contemporary mathematics. The easiest chaos game to understand is played as follows. Start with three points at the vertices of an equilateral triangle. Color one vertex red, one green, and one blue. Take a die and color two sides red, two sides green, and two sides blue. Then pick any point whatsoever in the triangle, this is the seed. Now roll the die. Depending upon which color comes up, move the seed half the distance to the similarly colored vertex. Then repeat this procedure, each time moving the previous point half the distance to the vertex whose color turns up when the die is rolled. After a dozen rolls, start marking where these points land.}

Devaney goes on to say, when this process is  repeated  thousands of times, the pattern that emerges is one of the most famous fractals of all, the {\it Sierpinski triangle}. The Sierpinski triangle consists of three self-similar pieces, each of which is exactly one half the size of the original triangle in terms of the lengths of the sides.

\begin{figure}[h!]
\centering
\includegraphics[width=2.0in]{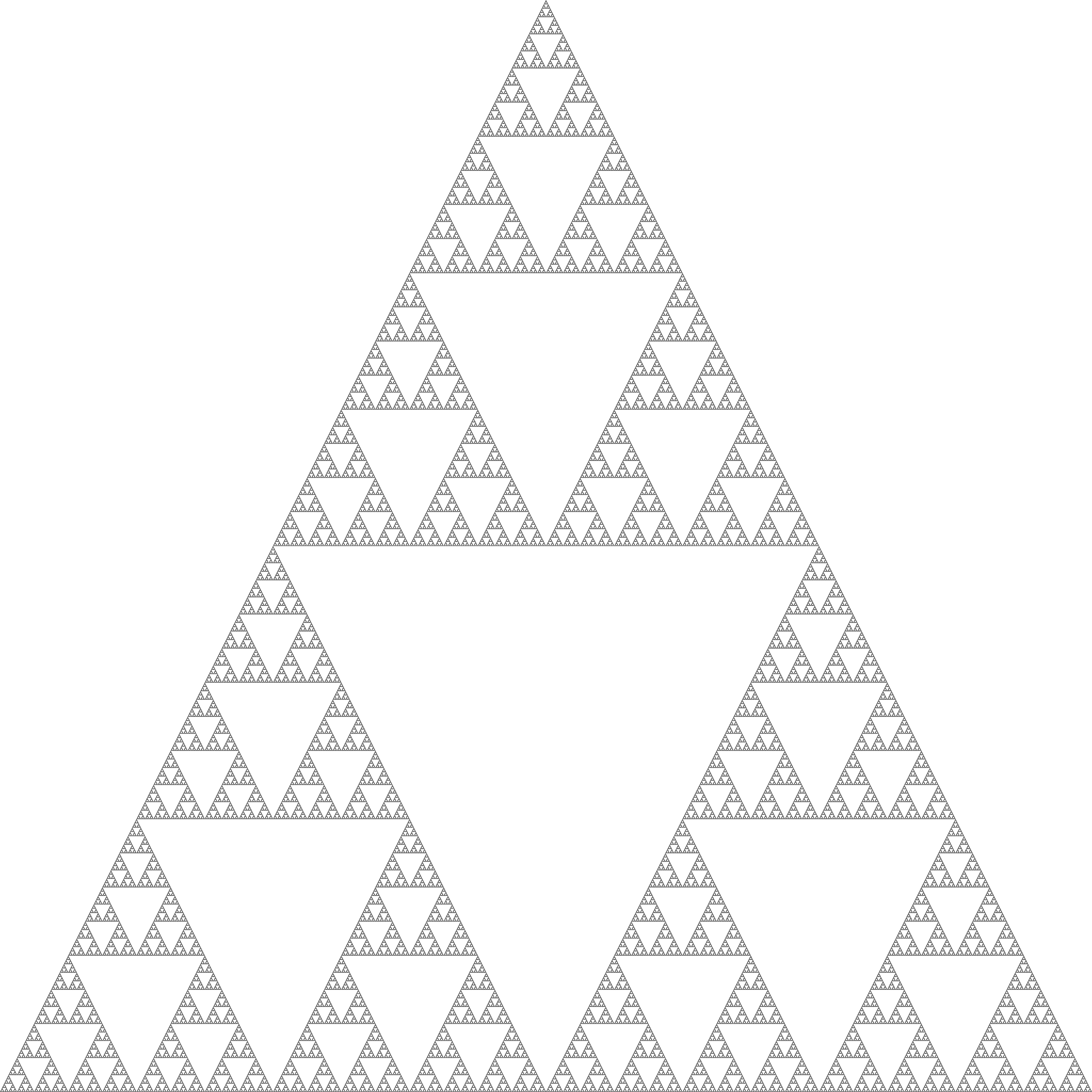}
\caption{The Sierpinski Triangle.} \label{Sier}
\end{figure}

\subsection{The Sierpinski-Randomized Triangle Algorithm}

Consider the convex hull problem for the very simple case where $S$ consists of three points as the vertices of an equilateral triangle and $p$ is a point inside the triangle.  We make the following claim on the Sierpinski triangle, see  Figure \ref{Sier} which is visually evident and provable from its topological properties. We  refer to the convex hull of the dots as enclosing Sierpinski triangle.\\

\begin{prop}
Given any dense subset of the Sierpinski triangle, $\Sigma$, no matter where $p$ is located inside the enclosing Sierpinski triangle, and no matter which of the three vertices is chosen as $v$, we can select a Sierpinski dot, say $p'$, for which the line segment $p'v$ either contains $p$, or comes as close to it as desired.
\end{prop}

The above gives an incentive to state a randomized triangle algorithm based on its generalization. First, consider the following generalization of the chaos game.

\begin{definition}(General Chaos Game)
Given a set of points  $S= \{v_1, \dots, v_n\} \subset \mathbb{R} ^m$, let $\Sigma(S)$ correspond to the dots generated via the following generalization of Sierpinski chaos game: Start with a seed $p' \in conv(S)$, and with probability $1/n$ randomly select $v \in S$, then record $(p'+v)/2$ as a new point and place it in $\Sigma(S)$. Replace $p'$ with $(p'+v)/2$ and repeat the process indefinitely.
\end{definition}

The following hypothesis gives the incentive to define another randomized triangle algorithm, what we call the {\it Sierpinski-Randomized Triangle Algorithm}.

\begin{hypothesis}
Suppose $p \in conv(S)$. Given $\varepsilon \in (0,1)$, $v \in S$,  there exists $p' \in \Sigma(S)$ such that

(i) $v$ is a $p$-pivot with respect to $p'$ (i.e. $d(p',v) \geq d(p, v)$),

(ii)  If $p''=nearest(p;p'v)$, then $d(p,p'') \leq \varepsilon d(p,v)$.
\end{hypothesis}

Regardless of the validity of the above hypothesis, we prove that bounds on the expected complexity of the Sierpinski-Randomized Triangle Algorithm is no worse than bounds on the worst-case complexity of the triangle algorithm itself.  The algorithm is inspired by the chaos game, however it keeps track of the current incumbent candidate, $p_*$, the closest known point to $p$.

Given an iterate $p'$, it randomly (with equal probability) selects $v \in S \cup \{p_*\}$. If $v$ is a pivot, it randomly  either replaces $p'$ with $(p'+v)/2$, or with $nearest(p;p'v)$. Otherwise, if $v \not = p_*$, the next iterate is $(p'+v)/2$, or else $v= p_*$ and not a pivot. In this case $p_*$ will be taken to be the iterate and the algorithm searches for a pivot $v'$ at $p_*$.  When such a pivot exists, the next iterate will be $nearest(p;p_*v')$. Except for this case, the other cases take $O(m+n)$ operations.

\begin{algorithm}[htpb]
{\bf Sierpinski-Randomized Triangle Algorithm}\

 \KwIn {$S= \{v_1, \dots, v_n\}$, $p$, $\varepsilon$}
 \KwOut {$p' \in conv(S)$, either $d(p, p') \leq \varepsilon R$,
 or  $p'$ a {\bf Witness}}
$p'={\rm argmin}\{d(p,v): v \in S\}$, $p_*=p'$\;
 \While{$(d(p,p')> \varepsilon R)$}{
  randomly select $v \in S \cup \{p_*\}$\;
  \eIf{$v$ is a $p$-pivot at $p'$}{at random set $p'=(p'+v)/2$, or $p' =nearest(p;p'v)$\;}
  {\eIf{$v=p_*$}{\eIf{no $p$-pivot exists at $p_*$}{$p'=p_*$, Output $p'$ a {\bf Witness} and halt\;}
  {given a $p$-pivot $v'$ at $p_*$, $p'=nearest(p;p_*v')$, $p_*=p'$\;}}
{$p'=(p'+v)/2$;}}
Update $p_*$;
 }
 Output $p'$\;
 %\caption{Sierpinski Randomized Triangle algorithms} \label{fig:Triangle}
\end{algorithm}

\begin{lemma} The expected number of arithmetic operations in each iteration of the Sierpinski-Randomized Triangle Algorithm is $O(m+n)$.
\end{lemma}

\begin{proof} The probability that in each iteration the randomly selected $v$ coincides with $p_*$ is $1/(n+1)$. Then if $v=p_*$ is not a $p$-pivot then $p_*$ becomes the new iterate and the number of operations to compute a pivot $v'$ at $p_*$ is $O(mn)$. If the randomly selected $v$ is not $p_*$, the number of operations to get the next iterate $p'$ is $O(m+n)$.  Thus the expected number or operations in each iteration is
$$\frac{n}{n+1} O(m+n)+ \frac{1}{n+1} O(mn) = O(m+n).$$
\end{proof}

\begin{thm} If $p \in conv(S)$, bound on the expected number of arithmetic operations to compute an $\varepsilon$-approximate solution is
$$O \bigg (\frac{mn}{\varepsilon^2} \bigg ).$$
Moreover, if it is known that $p$ is the center of ball of radius $\rho$ contained in the relative interior of $conv(S)$, and if each times it computes a pivot $v'$ for $p_*$ it is a strict pivot, bound on the expected number of arithmetic operations to compute an $\varepsilon$-approximate solution is
$$O\bigg (mn \bigg (\frac{R}{\rho} \bigg)^2 \ln \frac{1}{\varepsilon} \bigg).$$
\end{thm}
\begin{proof} The expected number of times a random $v$ is selected before it equals $p_*$ is $(n+1)$. When $p_*$ is not a pivot at the current iterate $p'$,  $p'$ is replaced with $p_*$ and a pivot $v'$ is computed.  Applying the existence results on  pivot and strict pivot, Theorem \ref{thm1} and Theorem \ref{thm2}, as well as the complexity bounds on the triangle algorithm, Theorems \ref{thm3} and \ref{thm4}, the proof follows.
\end{proof}

Suppose we consider a relaxed version of the above algorithm where each time $p_*$ is selected and is not a pivot at the current iterate, thus becoming a new iterate, we select $v'$ randomly and not necessarily as a pivot at $p_*$, thus economizing in computation.  Referring to this as the {\it Relaxed Sierpinski-Randomized Triangle Algorithm} we have.

\begin{thm} If $p \in conv(S)$, bound on the expected number of arithmetic operations of the Relaxed Sierpinski-Randomized Triangle Algorithm to compute an $\varepsilon$-approximate solution is
$$O \bigg (\frac{mn^2}{\varepsilon^2} \bigg ).$$
Moreover, if it is known that $p$ is the center of ball of radius $\rho$ contained in the relative interior of $conv(S)$, and if each times it computes a pivot $v'$ for $p_*$ it is a strict pivot, bound on the expected number of arithmetic operations to compute an $\varepsilon$-approximate solution is
$$O\bigg (mn^2 \bigg (\frac{R}{\rho} \bigg)^2 \ln \frac{1}{\varepsilon} \bigg).$$
\end{thm}
\begin{proof}  Each iteration takes $O(m+n)$ operations. Given an iterate $p'$,  the probability that a randomly selected $v$ in  $S \cup \{p_*\}$ is a $p$-pivot (strict $p$-pivot) at $p'$ is $1/(n+1)$. This is because the Voronoi cell of $p$ with respect to the two-point set $\{p,p'\}$ must contain a $p$-point in $S \cup \{p_*\}$ (otherwise, $p \not \in conv(S)$). Thus the probability that at an iterate $p_*$ is randomly selected and that at $p_*$ a pivot (strict pivot) is randomly selected is $1/n(n+1)$. From these and analogous arguments as in the previous theorem, the expected complexities follow.
\end{proof}

There is yet another relaxation: we treat $p_*$ as any other point in $S$, that is if an iterate $p'$ selects $p_*$ randomly,  we do not jump to $p_*$ as the next iterate. The expected complexity of this may remain to be the same as the relaxed version analyzed above.\\

{\bf Concluding Remarks.} In this article we have described randomized versions of the triangle algorithm.  Based on our previous theoretical and experimental  results,  see \cite{kal14}, \cite{Meng} and \cite{Gibson}, the triangle algorithm appears to be a promising algorithm with wide range of applications.  The randomized algorithms suggest variations that could help its performance in practice or in the worst-case.  Both allow exploring the the convex hull from different view points, thus increasing the chance to get better and better approximations to $p$ by choosing more effective pivots.  Also, in the Sierpinski-Randomized Triangle Algorithm as $p_*$ gets close to $p$, the chances are good that when an iterate $p'$ randomly selects $v=p_*$ that $p_*$ is actually a pivot at $p'$. Hence with probability $1/2$ the next iteration will get closer to $p$. To check if $p_*$ is a pivot at $p'$ takes $O(m+n)$ time as opposed to $O(mn)$ time. Additionally, it is likely that the randomized algorithms will help improve the performance of the triangle algorithm as a function of  $\varepsilon$.
Some theoretical questions thus arise. Computational experimentations are needed to assess practical values.  We plan to do so in future work.\\

{\bf Acknowledgements}  I like to thank Mike Saks for a discussion regarding randomization.

\bigskip

%\noindent\textit{Department of Computer Science, Rutgers University, Piscataway, NJ 08854\\
%kalantari@cs.rutgers.edu}


\begin{thebibliography}{9}
%\begin{thebibliography}{1}


\bibitem{Barn93} M. Barnsley, {\it Fractals Everywhere}, 1993,  Morgan Kaufmann.\filbreak

\bibitem{clark2008} K. L. Clarkson, Coresets, sparse greedy approximation,
and the Frank-Wolfe algorithm. In SODA
'08: Proceedings of the nineteenth annual ACM-SIAM symposium on Discrete algorithms,
922 - 931. Society for Industrial and Applied Mathematics,
2008. \filbreak

\bibitem{Devaney2004} R. L. Devaney, Chaos Rules!,
{\it Math Horizons},  (2004),  11-14. \filbreak

\bibitem{Frank} M. Frank and P. Wolfe, An algorithm for quadratic
programming, {\it Naval Res. Logist. Quart.},  3 (1956), 95 - 110. \filbreak

\bibitem{Gartner} B, G{\"a}rtner and M. Jaggi, Coresets for polytope distance,
Symposium on Computational Geometry (2009), 33 - 42.\filbreak

\bibitem{Gibson} T. Gibson and B. Kalantari, Experiments with the triangle algorithm for linear systems, 2-page Extended Abstract, 23nd Annual Fall Workshop on Computational Geometry, City College of New York, 2013. \filbreak

\bibitem{Gilbert} E. G. Gilbert, An iterative procedure for computing
the minimum of a quadratic form on a convex set, {\it
SIAM Journal on Control}, 4 (1966), 61 - 80. \filbreak

\bibitem{Goodman} J. E. Goodman, J. O'Rourke (Editors), {\it Handbook of Discrete and Computational Geometry}, 2nd Edition (Discrete Mathematics and Its Applications) 2004, Chapman \& Hall Boca Raton.\filbreak

\bibitem{kal14} B. Kalantari, A characterization theorem and an algorithm for a convex hull problem, to appear in {\it Annals of Operations Research}, available online August, 2014. \filbreak

\bibitem{Kalan12} B. Kalantari, Finding a lost treasure in convex hull of points from known distances. In the  Proceedings of the 24th Canadian Conference on Computational Geometry (2012), 271 - 276. \filbreak

\bibitem{kal12a} B. Kalantari, Solving linear system of equations via a convex hull algorithm, arxiv.org/pdf/1210.7858v1.pdf,  2012. \filbreak

\bibitem{kalSaks} B. Kalantari and M. Saks, On the triangle algorithm for the convex hull membership, 2-page Extended Abstract, 23nd Annual Fall Workshop on Computational Geometry, City College of New York, 2013.  \filbreak

\bibitem{kar84} N. Karmarkar,  A new polynomial time algorithm for linear programming, {\it Combinatorica}, 4 (1984),  373 - 395. \filbreak


\bibitem{kha79} L. G. Khachiyan, A polynomial algorithm in linear programming, {\it Doklady Akademia Nauk SSSR},
(1979), 1093 - 1096.\filbreak


\bibitem{Meng} M. Li and B. Kalantari, Experimental study of the convex hull decision problem via a new geometric algorithm, 2-page Extended Abstract, 23nd Annual Fall Workshop on Computational Geometry, City College of New York, 2013. \filbreak

\bibitem{zhang}   T. Zhang, Sequential greedy approximation for certain convex optimization problems, IEEE Trans. Information Theory, 49 (2003), 682 - 691. \filbreak

\end{thebibliography}
\end{document}